\documentclass[letterpaper, 10 pt, conference]{ieeeconf}  
\usepackage{amsmath}
\usepackage{amsfonts}
\usepackage{amssymb}
\usepackage{mathtools}
\usepackage{color}
\IEEEoverridecommandlockouts                              
\newtheorem{theorem}{Theorem}
\newtheorem{assumption}[theorem]{Assumption}
\newtheorem{definition}[theorem]{Definition}

\newtheorem{proposition}[theorem]{Proposition}

\newtheorem{corollary}[theorem]{Corollary}
\newcommand\RE{\mathbb{R}}

\title{Dominance analysis of linear complementarity systems}
\author{F. A. Miranda-Villatoro$^{*}$ \and F. Forni$^{*}$ \and R. Sepulchre$^{*}$
\thanks{The research leading to these results has received funding from the European Research Council
under the Advanced ERC Grant Agreement Switchlet n. 670645.}
\thanks{$^{*}$University of Cambridge, Department of Engineering, Trumpington Street, CB2 1PZ,  Cambridge, United Kingdom. Emails:
{\tt\small fam48@cam.ac.uk, f.forni@eng.cam.ac.uk, r.sepulchre@eng.cam.ac.uk}}}
\begin{document}
\maketitle
\thispagestyle{empty}
\pagestyle{empty}
\begin{abstract}
The paper extends the concepts of dominance and $p$-dissipativity to the  non-smooth family of  linear
complementarity systems. Dominance generalizes incremental stability whereas $p$-dissipativity generalizes
incremental passivity. The generalization aims at an interconnection theory for the design and analysis of switching and oscillatory systems.
The approach  is illustrated by a detailed study of classical electrical circuits that switch and oscillate.
\end{abstract}
\section{INTRODUCTION}
Dominance analysis and p-dissipativity were recently introduced in \cite{forni2017, Forni2017b}
to extend the application of dissipativity theory to the analysis of multistable and
oscillatory systems. The approach is {\it differential}, that is, based on the analysis of {\it linearized} dynamics
along trajectories, in the spirit of contraction theory \cite{lohmiller1998}, convergence analysis \cite{pavlov2005}, or differential stability
analysis \cite{forni2014b}. It is particularly adapted to systems whose linearization can be easily parametrized, such
as Lure systems that interconnect linear time-invariant systems with static nonlinearities \cite{Miranda2017b}.
The present paper investigates how to extend this analysis to 
nonlinear circuits modeled as linear complementarity systems:   models that consist of  linear time-invariant systems augmented
with a static complementarity constraint. 
The modeling framework of linear complementarity systems has proven very 
useful to analyze systems whose nonlinear dynamical behavior arises from non-smooth constraints
 \cite{Brogliato2001, Brogliato2003, Camlibel2002,
Camlibel2006, Heemels2003}. They find applications in a number of fields including mechanical 
systems with unilateral constraints \cite{Brogliato2016},  electrical circuits with
diodes \cite{Acary2011}, and mathematical programming \cite{nagurney1993}.
Linear complementarity systems provide an attractive framework for dominance analysis 
because they are general enough to model  switching and  oscillatory behaviors often encountered
in the presence of non-smooth constraints and specific enough to lead to tractable analysis.
In particular, the {\it passivity} property of complementarity constraints has proven central to
analyze linear complementarity systems in the framework of dissipativity 
theory  \cite{Schaft1996,Camlibel2002}. Linear complementarity systems hence offer an ideal
platform for the application of $p$-dissipativity theory to switching and oscillatory behaviors.
To account for the non-smoothness of linear complementarity systems, the {\it differential} analysis
of \cite{forni2017, Forni2017b} has to be replaced by {\it incremental} analysis \cite{stan2007}, \cite{pavlov2008}, \cite{Liu2011}. 
Incremental analysis studies how {\it increments} between trajectories evolve in time, whereas
differential analysis only considers {\it linearized} trajectories, that it, infinitesimal increments.
In the context of linear complementarity systems, the difference is technical rather than conceptual. 
We show that the main results of \cite{forni2017, Forni2017b} extend to the incremental setting imposed
by non-smooth constraints.
This paper is organized as follows. In Section \ref{section:lcs} we briefly review the modeling of 
linear complementarity systems and its core passivity property. Section \ref{section:dominance}
is dedicated to  the property of dominance and $p$-dissipativeness in the incremental setting.
The example Section \ref{section:example} illustrates the potential of $p$-dissipativity theory to analyze
classical switching and oscillatory circuits.
The papers ends with conclusions in Section \ref{section:conclusions}.
\section{Linear complementarity systems and incremental passivity}
\label{section:lcs}
\subsection{Linear complementarity systems}
A linear complementarity system \cite{Schaft1996} consists
of a linear dynamical system subject to a complementarity constraint
\begin{equation}
	\begin{cases}
		\dot{x} = A x + B u + Bv  
		\\
		y = C x + D u
		\\
		0 \leq u \perp y \geq 0 , 
	\end{cases}
	\label{eq:lcp}
\end{equation}
where $x \in \RE^{n}$ is the state variable, $u \in \RE^{m}$ and $y \in \RE^{m}$ are the 
so-called complementary variables, and $v \in \RE^{m}$ is an additional control input.
The matrices $A ,B, C$, and $D$ are constant and of the appropriate dimensions.
The complementarity condition $0 \leq u(t) \perp y(t) \geq 0$
is a compact representation of the following three
conditions: i) $u \in \RE_{+}^{m}$, 
ii) $y \in \RE_{+}^{m}$, and iii) $\langle u(t), y(t) \rangle = 0$. 
A solution of the linear complementarity system \eqref{eq:lcp} is any tuple $(x, u, y,v)$ such that  
$x: \RE_{+} \to \RE^{n}$ is an absolutely continuous function and $(x,u,y,v)$ satisfies \eqref{eq:lcp} 
for almost all forward times $t \in \RE_{+}$. 
In general, we assume that the initial conditions $x(0) = x_{0}$ are such that the complementarity conditions hold.
This implies the absence of jumps in the initial condition and in the complementarity variables \cite{Camlibel2002}.
\subsection{Incremental passivity}
A cornerstone in the analysis of linear complementarity systems is to observe that the complementarity relation 
	\begin{equation}
		R_{\perp} = \left\{(y, \zeta) \in \RE^{m} \times \RE^{m} | 0 \leq y \perp -\zeta \geq 0 \right\}
		\label{eq:complementarityRelation}
	\end{equation}
defines an incrementally passive relation. We recall the definition and the proof of that property.
\begin{definition}
	A relation $R\subseteq \RE^{m} \times \RE^{m}$ is incrementally passive if for any 
	$(y_{1},\zeta_{1}) \in R$, and any $(y_{2}, \zeta_{2}) \in R$ the inequality
	\begin{equation}
	\langle y_{1} - y_{2}, \zeta_{1} - \zeta_{2} \rangle \geq 0
	\label{eq:static:incPassive}
	\end{equation}
	holds.
\end{definition}
\begin{proposition}
	$R_{\perp}$  is incrementally passive.
\label{prop:complementarity:passive}
\end{proposition}
\begin{proof}
	Take $(y_{1}, \zeta_1)\in R_{\perp}$ and $(y_{2}, \zeta_2)\in R_{\perp}$. Then, 
	$
	\langle y_{1} - y_{2}, \zeta_{1} - \zeta_{2} \rangle = 
	\langle -y_{2} , \zeta_{1}  \rangle + \langle y_{1} , -\zeta_{2}  \rangle + \underbrace{\langle y_{1} , \zeta_{1}  \rangle}_{=0} + \underbrace{\langle y_{2} , \zeta_{2}  \rangle}_{=0} 
	= \underbrace{\langle y_{2} , -\zeta_{1}}_{\geq 0}  \rangle + \underbrace{\langle y_{1} , -\zeta_{2}  \rangle}_{\geq 0}  \geq 0
	$. \vspace{1mm}
\end{proof}
An alternative description of the  relation $R_{\perp}$ is via  the multivalued map
	\begin{equation} 
	\zeta = \varphi_{\!\perp}(y) \in \left\{ \zeta \in \RE^{m}  | 0 \leq y \perp -\zeta \geq 0 \right\}
	\end{equation}
which leads to the  feedback representation in Figure \ref{fig:lure:complementarity}.
\begin{figure}[htpb]
	\centering
	\includegraphics[width=0.25\textwidth]{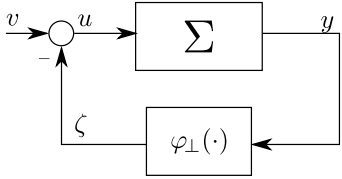}
	\caption{Block diagram of the linear complementarity system \eqref{eq:lcp}}
	\label{fig:lure:complementarity}
\end{figure}
This Lure type representation of linear complementarity systems
calls for an analysis rooted in passivity theory: the linear complementary system is incrementally passive as the negative feedback interconnection of a linear passive
system $\Sigma$ with an incrementally passive relation.
Passivity of the linear part is a standard assumption in the literature on
linear complementarity systems; it guarantees existence and uniqueness of solutions for \eqref{eq:lcp}. 
Details can be found in  \cite{Camlibel2002}, \cite{Brogliato2003}, \cite{miranda2017} based on the following additional
assumption, which ensures well-posedness of the closed loop in the presence of the throughput term $D$, \cite{Camlibel2002}. 
\begin{assumption}
	The linear part of \eqref{eq:lcp} is a minimal realization, it is passive and the matrix
	\begin{equation}
		\begin{bmatrix}
			B \\ D + D^{\top}
		\end{bmatrix}
	\end{equation}
	has full column rank.  \hfill $\lrcorner$
	\label{assumption:rank}
\end{assumption}
Passivity of the linear part of \eqref{eq:lcp} from $u$ to $y$ reads
\begin{equation}
\dot{V}(x) \leq \langle y, u\rangle
\label{eq:lin_pass}
\end{equation}
where the derivative of the quadratic storage $V:=x^T P x$, $P=P^T > 0$,
is computed along the linear dynamics.
Passivity and incremental passivity coincide for linear systems.
In fact, for any pair of trajectories $x_i$, outputs $y_i$ and inputs $u_i$,
the incremental dynamics characterized by the variables
$\Delta x = x_1-x_2$, $\Delta y = y_1-y_2$, $\Delta u = u_1-u_2$
satisfies
\begin{equation}
\dot{V}(\Delta x) \leq \langle \Delta y, \Delta u\rangle \ .
\label{eq:lin_incr_pass}
\end{equation}
Since the negative feedback interconnection of incrementally passive systems is incrementally passive \cite{Schaft1996},
the closed loop linear complementarity system is incrementally passive. 
For any constant input $v$, the resulting closed loop is thus incrementally stable, that is,
there exists a nondecreasing function $\beta$ such that
$$
|x_1(t) - x_2(t)| \leq \beta(|x_1(0)-x_2(0)|) \qquad \forall t \geq 0 
$$
for any pair of trajectories $x_1(\cdot)$, $x_2(\cdot)$ of \eqref{eq:lcp}. 
Furthermore, if the passive inequality \eqref{eq:lin_pass} is strict, the resulting 
closed loop becomes incrementally asymptotically stable; its trajectories
converge towards each other
$$
\lim_{t\to\infty} |x_1(t) - x_2(t)| = 0 \ .
$$
In this case, an equilibrium point is necessarily unique and globally stable. 
The concept of incremental stability \cite{angeli2000} is analog to  the concept of differential stability in  the theory of contraction \cite{lohmiller1998} or convergent systems
\cite{pavlov2005}. But it does not require any differentiability of the system dynamics. 
\subsection{Beyond linear complementarity relations}
\label{sec:beyond}
For the purpose of this paper, the linear complementarity condition 
can be replaced by \emph{any} incrementally passive static relation.
Figure \ref{fig:staticMaps} provides an illustration of incrementally passive memoryless 
nonlinearities. 
Those multivalued maps are widely used for modeling electronic circuits. 
For example, the three graphs in Figure \ref{fig:staticMaps}
represent the ideal current-voltage characteristic of a
diode, of a zener diode, and of an array of diodes \cite{Acary2010}, \cite{lootsma1999}.
\begin{figure}[htpb]
	\centering
	\includegraphics[width=0.35\textwidth]{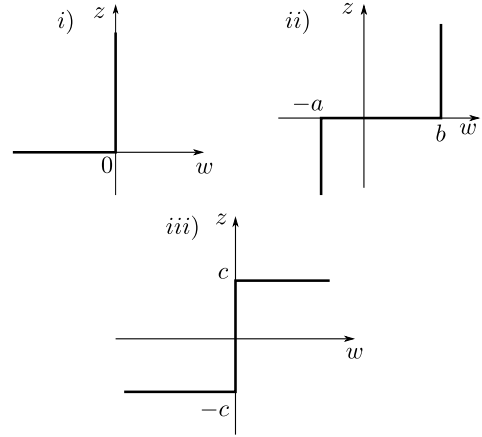}
	\caption{Popular examples of incrementally passive relations.}
	\label{fig:staticMaps}
\end{figure}
Denoting $(w,z) \in R_i$ any pair $(w,z)$ that belongs to the $i$-th relation in Figure \ref{fig:staticMaps},
we extend the class of linear complementarity systems to the family of systems of the form
\begin{equation}
	\begin{cases}
	\dot{x} = A x + B u + Bv
	\\
	y = C x
	\\
	(y,-u) \in R_i \ .
	\end{cases}
		\label{eq:diff:inclusion}
\end{equation}
Following the approach of linear complementarity systems, 
for a solution of \eqref{eq:diff:inclusion} we mean any tuple $(x,u,y,v)$ such that  
$x: \RE_{+} \to \RE^{n}$ is an absolutely continuous function and $(x,u,y,v)$ satisfies \eqref{eq:diff:inclusion} 
for almost all forward times $t \in \RE_{+}$. 
Indeed, the closed loop \eqref{eq:diff:inclusion} allows for the
block diagram representation in Figure \ref{fig:lure:complementarity},
with $\varphi_{\perp}$ replaced by $\varphi_{R_i}$, 
the static multivalued map associated to the relation $R_i$.
Thus, for any passive relation $R_i$, 
the closed loop \eqref{eq:diff:inclusion} is incrementally passive.
We remark that incremental passivity 
guarantees existence and uniqueness of solutions 
also for general \emph{maximal monotone} static multi-valued maps,
\cite{Brogliato2003}. 
\section{Dominance and $p$-dissipativity}
\label{section:dominance}
\subsection{Dominance}
Dominance was recently introduced in \cite{forni2017}, \cite{Forni2017b}, \cite{Miranda2017b} as a generalization
of incremental stability for smooth nonlinear systems. Motivated by dominance analysis of linear complementarity systems,
we extend the definition of dominance in a nonsmooth setting,
replacing {\it differential} analysis by {\it incremental} analysis as in the previous section.
For the sake of simplicity in this section, and the rest of the paper, we consider generic pairs of trajectories $x_1(\cdot)$ and $x_2(\cdot)$,
and we adopt the notation $\Delta x = x_1 - x_2$ to denote their mismatch. $\Delta \dot{x} = \dot{x}_1 - \dot{x}_2$ is defined for almost
every $t$ by the right-hand side of \eqref{eq:diff:inclusion} computed for $x_1$ and $x_2$, respectively. A similar notation is adopted
for inputs $\Delta u = u_1 -u_2$ and outputs $\Delta y = y_1 - y_2$. Finally, we say that a symmetric matrix $P$ has inertia $\{p, 0, n-p \}$ 
when it has $p$ negative eigenvalues and $n-p$ positive eigenvalues.
\begin{definition} \label{definition:idominance}
	The nonsmooth system  \eqref{eq:diff:inclusion}
	is $p$-dominant with rate $\gamma \geq 0$  if there exist a matrix $P = P^{\top}$ with inertia $\{p, 0, n-p \}$ and a constant 
	$\varepsilon \geq 0$ such that for any pair trajectories of \eqref{eq:diff:inclusion},
	\begin{equation}
		\begin{bmatrix}
			\Delta\dot{x} \\ \Delta{x}
		\end{bmatrix}\\
		\begin{bmatrix}
			0 & P 
			\\
			P & 2 \gamma P + \varepsilon I
		\end{bmatrix}
		\begin{bmatrix}
			\Delta\dot{x} \\ \Delta{x}
		\end{bmatrix} \leq 0.
		\label{eq:idominance}
	\end{equation}
	Strict $p$-dominance holds for $\varepsilon > 0$. 
\end{definition}
Note that dominance is just  incremental stability if $P$ is positive definite, which corresponds to $p=0$. But we are interested 
in the generalization corresponding to $p=1$ and to $p=2$.
For smooth closed systems  $\dot{x} = f(x)$,
\eqref{eq:idominance} is equivalent to the linear matrix inequality
inequality
	\begin{equation}
		 \partial f(x)^{\top} P + P \partial f(x) + 2
		\gamma P + \varepsilon I \leq 0 \qquad \forall x \in \RE^n
		\label{eq:lmi:dominance}
	\end{equation}
where $\partial f(x)$ denotes the Jacobian of $f$ at $x$. 
In the linear case, $f(x) = A x$, \eqref{eq:lmi:dominance} implies
the  existence of an invariant splitting such that $\RE^{n} = E_{n} \oplus E_{n-p}$, 
where $E_{p}$ is the $p$-dimensional eigenspace
associated to the \emph{dominant} modes of $A$ (eigenvalues of $A$ whose real part is larger than $-\gamma$),
and $E_{n-p}$ is the $(n-p)$-dimensional eigenspace associated
to the \emph{transient} modes of $A$ (i.e., the eigenvalues of $A$ whose real part is smaller than $-\gamma$).
Roughly speaking, in the nonlinear case the property of dominance forces the asymptotic behavior
to be $p$-dimensional \cite[Theorem 2]{Forni2017b}, as shown by the analysis of the linearized flow in \cite[Theorem 1]{Forni2017b}.
The following theorem extends \cite[Theorem 2]{Forni2017b} to the nonsmooth case. 
\begin{theorem}
	Assume $v$ constant and suppose that all the trajectories of $\eqref{eq:inclusion}$ are bounded.
	Let $\Omega(x)$ be the set of all $\omega$-limit points of $x$ and
	let \eqref{eq:inclusion} be strictly $p$-dominant with rate $\gamma \geq 0$. 
	Then, the flow on the $\Omega(x)$ is
	topologically equivalent to the flow of a $p$-dimensional system.
	\label{theorem:pdominance:behaviour}
\end{theorem}
\begin{proof}
	Consider any pair of trajectories $x_1(\cdot)$ and $x_2(\cdot)$,
	define the increment $\Delta{x}(\cdot)$, and consider the quadratic form $V(\Delta x(t)) = \Delta x(t)^T P \Delta x(t)$.
	From \eqref{eq:idominance},
	\begin{displaymath}
		\frac{d}{dt} V(\Delta{x}(t))  \leq -2 \gamma V(\Delta{x}(t)) - \varepsilon \Vert
		\Delta{x}(t) \Vert^{2},
	\end{displaymath}
	therefore
	\begin{displaymath}
		\frac{d}{dt} e^{2 \gamma t} V(\Delta{x}(t))  \leq -\varepsilon e^{2\gamma t} \Vert 
		\Delta{x}(t) \Vert^{2}.
	\end{displaymath}
	By integration, 
	\begin{equation}
		\label{eq:invariance}
		V(\Delta{x}(t)) \leq e^{-2 \gamma t} V(\Delta{x}(0)) - \varepsilon \int_{0}^{t}
		e^{-2 \gamma(t -\tau)} \Vert \Delta{x}(\tau) \Vert^{2} d \tau
	\end{equation}
	Let $\bar{x}_{1}$ and $\bar{x}_{2}$ be two different points of $\Omega(x)$
	and define $\Delta \bar{x} = \bar{x}_1 - \bar{x}_2$.
	Note that both $\bar{x}_{1}$ and $\bar{x}_{2}$ are accumulation points of a suitable trajectory, 
	therefore $\Delta \bar{x} \neq 0$ and \eqref{eq:invariance} implies 
	\begin{equation}
		V(\Delta_{\bar{x}}) <0.
		\label{eq:invariance2}
	\end{equation}	
	Let $\mathcal{H}_{P}$, $\mathcal{V}_{P}$ be the eigenspaces of $P$ associated to the
	$p$ negative eigenvalues of $P$, and $n-p$ positive eigenvalues of $P$, respectively.
	Let $\Pi: \RE^{n} \to \mathcal{H}_{P}$ be the projection onto $\mathcal{H}_{P}$ along $\mathcal{V}_{P}$.
	We claim that $\Pi$ restricted to ${\Omega}$ is one-to-one. In fact, assume by contradiction that for
	$\bar{x}_{1}, \bar{x}_{2} \in \Omega(x)$, $\bar{x}_{1} \neq \bar{x}_{2}$ implies 
	$\Pi(\Delta_{\bar{x}}) = 0$, it follows that $\Delta_{\bar{x}} \in \mathcal{V}_{P}$
	and therefore $V(\Delta_{\bar{x}}) > 0$ which contradicts \eqref{eq:invariance2}. Hence,
	$\Pi$ restricted to ${\Omega(x)}$ is one-to-one. 
	
	Now, for each constant input $v$, consider the equivalent representation of the system \eqref{eq:diff:inclusion} based on
	the differential inclusion  
	\begin{equation}
	\dot{x} \in \mathbf{F}_v(x) \ .
	\label{eq:inclusion}
\end{equation} 
	Using the results above, if $y \in \Pi \Omega(x)$ then there exists a unique
	initial condition $z(0) \in \Omega(x)$ such that $y = \Pi z(0)$ and the flow $\Pi z(t)$
	in $\mathcal{H}_{P}$ is generated by the vector field
	\begin{equation}
		\mathbf{G}_v(y) = \Pi \mathbf{F}_v (\Pi^{-1} y), \quad y \in \Omega(x)
\end{equation} 
	which is $p$-dimensional.
\end{proof}
Theorem \ref{theorem:pdominance:behaviour} shows that the asymptotic behavior of 
a strict $p$-dominant system is strongly constrained for small values of $p$. 
\begin{corollary}
	\label{corollary:behaviour}
	Let the assumptions of Theorem \ref{theorem:pdominance:behaviour} hold. In addition,
	assume that solutions of  \eqref{eq:diff:inclusion} are unique. Then all solutions asymptotically converge
	to
	\begin{enumerate}
		\item a unique equilibrium point, if $p = 0$.
		\item an equilibrium point, if $p =1$.
		\item an equilibrium point, a set of equilibrium points and connecting arcs, or a
			limit cycle, if $p = 2$.
	\end{enumerate}
\end{corollary}
Under uniqueness of solutions, distinct trajectories cannot intersect.
For $p=1$, the asymptotic dynamics are one-dimensional,  forcing bounded trajectories
to converge to some fixed point. Uniqueness of solutions is also sufficient for guaranteeing the validity of the 
Poincar{\'e}-Bendixson Theorem, see e.g., \cite[Theorem 5.3]{ciesielski1994}. 
Hence, under the assumption of uniqueness of solutions, a $2$-dominant system with a compact limit set that contains no equilibrium point
has a closed orbit.
\subsection{Incremental $p$-dissipativity}
Dissipativity theory is an interconnection theory for stability analysis.  In the same way,  $p$-dissipativity is an interconnection theory  for dominance analysis 
\cite{Forni2017b}, \cite{Miranda2017b}. It mimics standard dissipativity theory 
in the differential/incremental setting but relies on quadratic storage functions that have a prescribed inertia. 
\begin{definition}
	A nonsmooth system \eqref{eq:diff:inclusion} is $p$-dissipative with rate $\gamma \geq 0$ 
	and incremental supply $w: \RE^{m} \times \RE^{m} \to \RE$
	\begin{equation}
	w(\Delta{u}, \Delta{y}) :=
	\begin{bmatrix}
		\Delta{y} \\ \Delta{u}
	\end{bmatrix}^{\top}
	\begin{bmatrix}
		Q & L^{\top}
		\\
		L & R
	\end{bmatrix}
	\begin{bmatrix}
		\Delta{y} \\ \Delta{u}
	\end{bmatrix}
\end{equation}
	if there exists a matrix $P = P^{\top}$ with inertia $\{p, 0, n-p \}$ and $\varepsilon \geq0$ such that
	\begin{equation}
		\begin{bmatrix}
			\Delta\dot{x} \\ \Delta{x}
		\end{bmatrix}\\
		\begin{bmatrix}
			0 & P 
			\\
			P & 2 \gamma P + \varepsilon I
		\end{bmatrix}
		\begin{bmatrix}
			\Delta\dot{x} \\ \Delta{x}
		\end{bmatrix} \leq w(\Delta{y}, \Delta{u})
		\label{eq:inc:pdissipative}
	\end{equation}
	for any pair of trajectories.
	Strict $p$-dissipativity holds for $\varepsilon > 0$.
\end{definition}
$0$-dissipativity coincides with the classical concept  of incremental dissipativity.
$p$-dissipativity allows for an interconnection theory for
non-smooth systems, as clarified by the next theorem. 
The  simplest example is given by the closed loop in  
Figure \ref{fig:lure:complementarity}. 
\begin{theorem}
	\label{theorem:pdissipative}
	Let $\Sigma_{1}$ and $\Sigma_{2}$ be (strict) $p_{1}$ and 
	$p_{2}$ dissipative respectively, both with rate $\gamma \geq 0$ and supplies
	\begin{displaymath}
		w^i(\Delta{u}^{i}, \Delta{y}^{i}) = 
		\begin{bmatrix}
			\Delta{y}^{i} \\ \Delta{u}^{i}
		\end{bmatrix}^{\top}
		\begin{bmatrix}
			Q_{i} & L_{i}
			\\
			L_{i}^{\top} & R_{i}
	\end{bmatrix}
	\begin{bmatrix}
		\Delta{y}^{i} \\ \Delta{u}^{i}
	\end{bmatrix}, \quad i = 1, 2.
	\end{displaymath}
	The negative feedback interconnection 
	\begin{displaymath}
		u^{1} = -y^{2} + v^{1}, \quad u^{2} = y^{1} + v^{2}
	\end{displaymath}
	of $\Sigma_{1}$ and
	$\Sigma_{2}$ is (strict) $(p_{1} + p_{2})$-dissipative with respect to the input 
	$v := [v^{1}, v^{2}]^{\top}$ and the output $y := [y^{1}, y^{2}]^{\top}$,
	with incremental supply given by
	\begin{displaymath}
		\begin{bmatrix}
			\Delta y \\ \Delta v 
		\end{bmatrix}^{\top}
		\begin{bmatrix}
			Q_{1} + R_{2} & -L_{1} + L_{2}^{\top}& L_{1} & R_{2}
			\\
			-L_{1}^{\top} + L_{2} & Q_{2} + R_{1} & -R_{1} & L_{2}
			\\
			L_{1}^{\top} & -R_{1}^{\top} & R_{1} & 0
			\\
			R_{2} & L_{2}^{\top} & 0 & R_{2}
		\end{bmatrix}
		\begin{bmatrix}
			\Delta y \\ \Delta v 
		\end{bmatrix} \ .
	\end{displaymath}
	and rate $\gamma \geq 0$.
	In addition, if
	\begin{displaymath}
		\begin{bmatrix}
			Q_{1} + R_{2} & - L_{1} + L_{2}^{\top}
			\\
			-L_{1}^{\top} + L_{2} & Q_{2} + R_{1}
		\end{bmatrix} \leq 0
	\end{displaymath}
	then the interconnection is (strictly) $(p_{1} + p_{2})$-dominant.
	\label{eq:dissipativeConexion}
\end{theorem}
\begin{proof}
The proof follows by standard arguments of dissipativity theory. See also \cite{forni2017}.
\end{proof}
Classical dissipativity theory provides a tool to
analyze stable systems, that is, $0$-dominant systems, as interconnection of dissipative
open systems, that is,  $0$-dissipative systems. 
Theorem \ref{theorem:pdissipative} generalizes this conclusion: $1$-dominant systems
can be analyzed as interconnections of $0$-dissipative systems with a $1$-dissipative system; 
 $2$-dominant systems
can be analyzed as interconnections of $0$-dissipative systems with a $2$-dissipative system, or 
as interconnections of two $1$-dissipative systems.
\subsection{$p$-Passivity of linear complementarity systems}
As in the classical theory,  $p$-passivity is $p$-dissipativity for the particular supply rate
	\begin{equation}
	w(\Delta{u}, \Delta{y}) :=
	\begin{bmatrix}
		\Delta{y} \\ \Delta{u}
	\end{bmatrix}^{\top}
	\begin{bmatrix}
		0 & I
		\\
		I & 0
	\end{bmatrix}
	\begin{bmatrix}
		\Delta{y} \\ \Delta{u}
	\end{bmatrix}
	\label{eq:p-passivity}
\end{equation}
We have seen in Section \ref{section:lcs}  that the linear complementarity relation
$R_{\!\perp}$ is $0$-passive (since $0$-passivity and incremental passivity coincide).
Also the static nonlinearities $R_i$ in Section \ref{sec:beyond} are $0$-passive. 
Thus, from Theorem \ref{theorem:pdissipative}, 
the non-smooth system \eqref{eq:diff:inclusion} is the negative feedback loop of a $0$-passive
static nonlinearity with a linear system, whose degree of $p$-passivity determines
the degree of passivity of the closed loop. The degree of passivity of the linear part 
restricts the asymptotic behavior of the system.
We observe that for linear systems of the form $\dot{x} = Ax + Bu$, $y = Cx$,
the inequality \eqref{eq:inc:pdissipative} with \eqref{eq:p-passivity} reduces to the simple feasibility test
$$
A^T P + PA + 2 \gamma P \leq -\varepsilon I \qquad PB = C^T
$$
for some matrix $P=P^T$ with inertia $\{p,0,n-p\}$; a numerically tractable condition.
Also, $p$-passivity has a frequency domain characterization
based on the rate-shifted transfer function $G(s - \gamma) = C {(s I - (A + \gamma I))}^{-1} B$,
\cite{Miranda2017b}:
\begin{proposition}
	\label{proposition:frequency}
	A linear system is $p$-passive if and only if the following two conditions hold,
	\begin{enumerate}
		\item\label{cond:1} $\Re \left\{ G(j \omega - \gamma) \right\} > 0$, 
			for all, $\omega \in \RE \cup \{+\infty\}$.
		\item\label{cond:2} $G(s - \gamma)$ has $p$ poles on the right-hand side of the complex plane.
	\end{enumerate}
\end{proposition}
Frequency domain conditions prove useful
in capturing the limits of the theory and for the selection of the systems parameters,
as shown in the next section.
\section{Switching and oscillating LCS circuits}
\label{section:example}
\subsection{The operational amplifier is $0$-passive}
A model of the operational amplifier is 
shown in Figure \ref{fig:nonsmoothModel1}: a first order model \cite{karki2000}, with additional voltage saturation limits,
implemented via ideal diodes, to take into account the physical limitations of any op-amp device. 
Note that the right-most element $\times 1$ in the model denotes a buffer, 
which decouples the output current from the internal circuit. We make
 the  usual assumption of infinite input impedance $Z_{in} = + \infty$ and $0$ 
output impedance $Z_{out} = 0$.
\begin{figure}[htpb]
	\centering
	\includegraphics[width=0.35\textwidth]{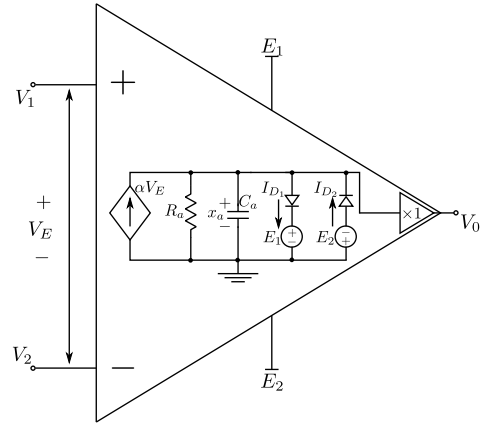}
	\caption{A nonsmooth model of an operational amplifier}
	\label{fig:nonsmoothModel1}
\end{figure}
The model in Figure \ref{fig:nonsmoothModel1} is described by the
linear complementarity system
\begin{subequations}
	\label{eq:opamp:Sys1}
	\begin{align}
		\dot{x}_{a} & = - \frac{1}{R_{a} C_{a}} x_{a} + \frac{\alpha}{C_{a}} V_{E} - \frac{1}{C_{a}} (I_{D_{1}}(t) - I_{D_{2}}) 		
		\label{eq:opamp:1}
		\\
		V_{0} & = x_{a}
		\label{eq:opamp:2}
		\\
		0 & \leq -V_{0} + E_{1} \perp I_{D_{1}} \geq 0,
		\label{eq:opamp:3}
		\\
		0 & \leq V_{0} + E_{2} \perp I_{D_{2}} \geq 0.
		\label{eq:opamp:4}
	\end{align}
\end{subequations}
where we assume that the voltage sources $E_{1}$ and $-E_{2}$ are constant and $E_{i} > 0$, 
$i =1,2$.  The model is derived via Kirchhoff's laws together with the complementarity conditions 
for the diodes \cite{Acary2010}. $x_{a}$ denotes the voltage across the capacitor $C_{a}$. The input is set to 
$V_{E}$ and the output is $V_{0}$.  
For completeness, we write
\eqref{eq:opamp:Sys1}   in the standard linear complementarity form 
\eqref{eq:lcp} by taking $V_{E} = V_{1} - V_{2}$, $u = [ I_{D_{1}}, I_{D_{2}}, E_{1}, E_{2} ]^{\top}, 
v = [ -\alpha V_{1}, -\alpha V_{2},  0,  0]^{\top}$,
$A = -\frac{1}{R_{a}C_{a}}$, $B = \frac{1}{C_{a}} [-1, 1, 0, 0]$, 
\begin{displaymath}
	C = 
	\begin{bmatrix}
		-1 \\ 1 \\ 0 \\ 0 
	\end{bmatrix}, \text{ and }
	D = 
	\begin{bmatrix}
	0 & 0 & 1 & 0
	\\
	0 & 0 & 0 & 1
	\\
	0 & 0 & 0 & 0
	\\
	0 & 0 & 0 & 0
	\end{bmatrix}. 
\end{displaymath}  
A representation of \eqref{eq:opamp:Sys1} is the negative feedback interconnection of a strictly $0$-passive system
with a $0$-passive relation $R$ defined by \eqref{eq:opamp:2},\eqref{eq:opamp:3}, and \eqref{eq:opamp:4},
linking voltage $V_{0}$ and the difference of diode currents $I_{DD} = I_{D_{1}} - I_{D_{2}}$.
The relations $0 \leq - V_{0}(t) + E_{1}$ and 
$0 \leq V_{0}(t) + E_{2}$ imply $-E_{2} \leq V_{0}(t) \leq E_{1}$. Thus, the output of 
the op-amp $V_{0}(t) = x_{a}(t)$ always belongs to the interval $[-E_{2}, E_{1}]$. 
For $V_{0}(t) = E_{1}$,  \eqref{eq:opamp:3} and \eqref{eq:opamp:4} imply 
that $I_{D_{1}}(t) \geq 0$ and $I_{D_{2}}(t) = 0$, respectively, that is, $I_{DD}(t) \geq 0$.
Similarly, $V_{0}(t) = -E_{2}$ implies $I_{DD}(t) \leq 0$ and $V_{0}(t) \in (-E_{2}, E_{1})$ implies $I_{DD}(t) = 0$. 
Hence, $R$ corresponds to the $0$-passive relation represented in Figure \ref{fig:staticMaps}.ii). Its 
associated multivalued function $\varphi_R$ maps the voltage $V_0$ into
\begin{displaymath}
	I_{DD} = \varphi_R(V_0) \in
	\begin{cases}
	(-\infty, 0], & V_{0}(t) = -E_{2}
		\\
		\{0\}, &  -E_{2} < V_{0}(t) < E_{1}
		\\
		[0, +\infty), & V_{0}(t) = E_{1} \ .
	\end{cases}
\end{displaymath}
The operational amplifier model \eqref{eq:opamp:Sys1} is thus given by the feedback loop in 
Figure \ref{fig:openOpamp}, combining the $0$-passive linear system $\Sigma_a$ with matrices
$A = -\frac{1}{R_{a} C_{a}}$, $B = \frac{1}{C_{a}} $, and $C = 1$, with the $0$-passive
nonlinearity $\varphi_R$. The transfer function of $\Sigma_a$ reads 
\begin{equation}
	G(s) = \frac{\frac{1}{C_{a}}}{s + \frac{1}{R_{a} C_{a}}}.
\end{equation}
and Proposition \ref{proposition:frequency} guarantees strict $0$-passivity with rate $\gamma \in [0, \frac{1}{R_{a} C_{a}})$.	
By Theorem \ref{theorem:pdissipative}, the op-amp is thus a strictly $0$-passive device from $V_E$ to $V_0$ with rate $\gamma \in [0, \frac{1}{R_{a} C_{a}})$. 
\begin{figure}[htpb]
	\centering
	\includegraphics[width=0.3\textwidth]{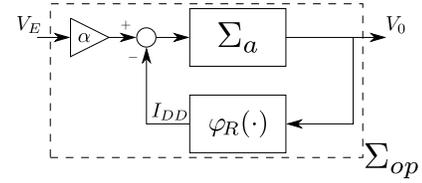}
	\caption{Block diagram of the operational amplifier
		\eqref{eq:opamp:Sys1}}
	\label{fig:openOpamp}
\end{figure}
\subsection{Positive feedback amplifier: multistable Schmitt trigger}
The positive feedback interconnection of 
the op-amp with an additional passive network leads to $1$-passive circuits. For instance, 
the Schmitt trigger circuit represented in Figure \ref{fig:schmittTrigger}
contains an op-amp, whose model is given by \eqref{eq:opamp:Sys1}, and a linear
network $\Sigma_c$ represented by
\begin{equation}
	\Sigma_{c}: \begin{cases}
		\dot{x}_{1} = - \frac{R_{1} + R_{2}}{R_{1} R_{2} C_{1}} x_{1} + \frac{1}{ R_{2} C_{1}} \nu_{1}\\
		y_{1}(t) = -x_{1} \ ,
	\end{cases}
	\label{eq:network:1}
\end{equation}
where $x_{1}$ is the voltage across the capacitor $C_{1}$.
\begin{figure}[htpb]
	\centering
	\includegraphics[width=0.25\textwidth]{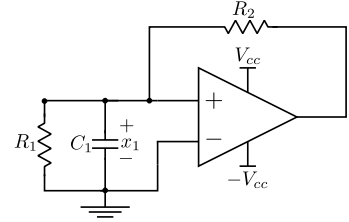}
	\caption{Schmitt Trigger circuit formed as the positive feedback interconnection
		of the circuit in Figure \ref{fig:nonsmoothModel1} with a RC network}
	\label{fig:schmittTrigger}
\end{figure}
Their interconnection is characterized by the positive feedback identity
\begin{equation}
	V_{E} = x_{1}, \quad \nu_{1} = V_0 \ .
	\label{eq:positiveFeedback}
\end{equation}
Positive feedback loops of two $0$-passive systems are not $0$-passive. But
$\Sigma_c$ is also strictly $1$-passive from the input $\nu_1$ to the output $y_1$ with rate 
$\gamma \in (\frac{R_{1} + R_{2}}{R_{1}R_{2}C_{1}}, + \infty)$,
which allows to rewrite \eqref{eq:positiveFeedback} as negative feedback
\begin{equation}
	V_{E} = - y_1, \quad \nu_{1} = V_0 \ .
	\label{eq:positive_as_negativeFeedback}
\end{equation}
Indeed, the positive feedback in \eqref{eq:positiveFeedback} is equivalent to the negative
feedback between a strictly $0$-passive system, the op-amp, and a strictly $1$-passive system, $\Sigma_c$.
Thus, by selecting the circuit parameters to satisfy
\begin{equation}
	\frac{R_{1} + R_{2}}{R_{1}R_{2}C_{1}} < \frac{1}{R_{a} C_{a}},
		\label{eq:schmittTrigger:condition:1}
\end{equation}
Theorem \ref{theorem:pdissipative} guarantees that the Schmitt trigger is strictly $1$-passive 
with rate $\gamma \in (\frac{R_{1} + R_{2}}{R_{1}R_{2}C_{1}}, \frac{1}{R_{a} C_{a}})$.
The closed loop is thus strictly $1$-dominant.
An interpretation of \eqref{eq:schmittTrigger:condition:1} is that the linear circuit $\Sigma_c$ 
must have a slower dynamics than the op-amp dynamics, to determine a dominant behavior of dimension $1$.
Mathematically, \eqref{eq:schmittTrigger:condition:1} guarantees the existence of a common rate $\gamma \geq 0$
for which op-amp and $\Sigma_c$ are respectively $0$-passive and $1$-passive, as required by Theorem \ref{theorem:pdissipative}.
A $1$-passive circuit can be multistable. Bounded trajectories of a strictly $1$-dominant system necessarily converge
to a fixed point. Boundedness of trajectories follows from the saturation of the op-amp voltage, which essentially ``opens the loop'' 
for large overshoots. To enforce multistability, we look for circuit parameters that guarantee the existence
of at least one unstable equilibrium point.  The condition 
\begin{equation}
	\frac{1}{R_{a}} < \frac{\alpha R_{1}}{R_{1} + R_{2}},
	\label{eq:schmittTrigger:condition:2}
\end{equation}
 makes the zero-equilibrium unstable.
The parameters in Table \ref{tab:paramSchmitt} satisfies 
\eqref{eq:schmittTrigger:condition:1} and \eqref{eq:schmittTrigger:condition:2}, enabling bistability.
The value of resistance $R_a$ and capacitance $C_a$ have been taken from \cite{karki2000},
to ensure good matching between simulations and the behavior of a real op-amp component.
Figure \ref{fig:bistable} shows the trajectories of the system from two different initial conditions
$(x_a(0),x_1(0)) \in \{(-2,-2),(2,2)\}$. 
\begin{figure}[htpb]
	\centering
	\includegraphics[width=0.31\textwidth]{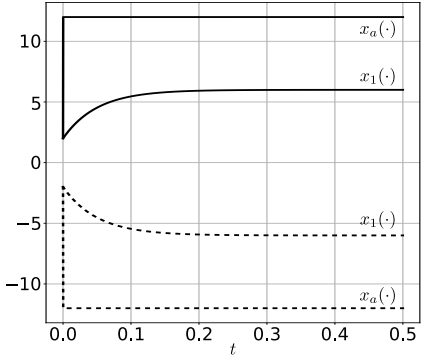} \vspace{-3mm}
	\caption{Schmitt trigger trajectories from two different initial conditions.}
	\label{fig:bistable}
\end{figure}
\begin{table}[htpb]
	\centering
	\begin{tabular}{|c|c|c|c|}
		\hline
		$R_{1} = 1 K \Omega$ & $R_{a} = 1 M \Omega$ & $C_{1} = 100 \mu F$ & $E_{1} = 12 V$
		\\
		\hline
		$R_{2} = 1 K \Omega$ & $\alpha = 0.1$ & $C_{a} = 15.9 nF$ & $E_{2} = 12 V$
		\\
		\hline
	\end{tabular}
	\caption{Parameter values for the simulation of the Schmitt trigger.}
	\label{tab:paramSchmitt}
\end{table}
\subsection{Mixed feedback amplifier: relaxation oscillator}
The interconnection of the op-amp with two slow passive networks,
one with negative and one with positive feedback, leads to $2$-passive circuits.
For instance, the circuit in Figure \ref{fig:relaxOscillator} is a typical architecture
for the generation of relaxation oscillations. It is 
derived from the Schmitt trigger through the addition of a slow network $\Sigma_d$, in negative feedback,
represented by 
\begin{equation}
	\Sigma_{d}:
	\begin{cases}
		\dot{x}_{2} = - \frac{R_{3} + R_{4}}{R_{3} R_{4} C_{2}} x_{2} + \frac{1}{R_{4} c_{2}}\nu_{2}\\
		y_{2} = x_{2} ,
	\end{cases}
	\label{eq:rc:network2}
\end{equation}
where $x_2$ is the voltage across the capacitor $C_2$.
\begin{figure}[htpb]
	\centering
	\includegraphics[width=0.25\textwidth]{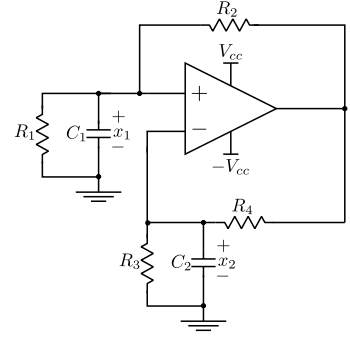}
	\caption{Relaxation Oscillator realized as the mixed positive and negative feedback
	of two RC networks and an op-amp}
	\label{fig:relaxOscillator}
\end{figure}
The interconnection of the op-amp and of the two networks $\Sigma_c$ and $\Sigma_d$
is given by the mixed positive/negative feedback
\begin{equation}
	V_{E} = x_1-x_2, \qquad \nu_{1} = \nu_{2} = V_0 \ .
\end{equation}
which can be written as a standard negative feedback loop using the aggregate output $y = y_1 + y_2$, 
for instance 
\begin{equation}
	V_{E} = - y = -y_1-y_2, \qquad \nu_{1} = \nu_{2} = V_0 \ .
\end{equation}
Taking $a_{1} = \frac{1}{R_{2} C_{1}} $, $a_{2} = \frac{1}{R_{4} C_{2}}$, $b_{1} = \frac{R_{1} + 
R_{2}}{R_{2} R_{1} C_{1}}$ and $b_{2} = \frac{R_{3} + R_{4}}{R_{3} R_{4} C_{2}}$,
the aggregate transfer function from $V_0$ to $y$ reads 
\begin{equation}
G(s)  = -\frac{a_1}{s + b_1} + \frac{a_2}{s + b_2}  = \frac{(a_{2} - a_{1})s + a_{2} b_{1} - a_{1} b_{2}}{(s + b_{1})(s + b_{2})} \ .
\end{equation}
For $a_2 = 0$, there is no negative feedback and the system  reduces to the Schmitt trigger. 
For $a_2 \neq 0$, the negative feedback loop either stabilizes the closed loop system, typically for
parameter values that guarantee $0$-passivity of $G(s)$, or induces oscillations, typically
for parameter values that guarantee $2$-passivity of $G(s)$. 
For instance, $G(j \omega -\gamma)$ has positive real part if 
\begin{align*}
	a_{2}(b_{2} - \gamma) - a_{1} (b_{1} - \gamma) > 0
	\\
	a_{2}(b_{1} - \gamma) - a_{1}(b_{2} - \gamma) > 0
\end{align*}
Hence, by Proposition \ref{proposition:frequency}, if
\begin{displaymath}
0 \!<\! \gamma \!<\! \min \!\left\{ \!b_{1}, b_{2}, \frac{1}{R_{a} C_{a}} \! \right\} ,\ 
	\frac{a_{2}}{a_{1}} \!>\! \max \!\left\{\! \frac{b_{1} - \gamma}{b_{2} - \gamma}, \!
	\frac{b_{2} - \gamma}{b_{1} - \gamma} \!\right\}
\end{displaymath}
then $G(s)$ is strictly $0$-passive with rate $\gamma$. 
The overall closed loop has a globally asymptotically stable fixed point. 
In contrast, for
\begin{displaymath}
0 \!<\! \max \!\left\{ b_{1}, b_{2}\right\} \!<\! \gamma \!<\! \frac{1}{R_{a} C_{a}}  \,,\ 
	\frac{a_{2}}{a_{1}} \!<\! \min \!\left\{\! \frac{b_{1} - \gamma}{b_{2} - \gamma}, \!
	\frac{b_{2} - \gamma}{b_{1} - \gamma} \!\right\}
\end{displaymath}
$G(s)$ is strictly $2$-passive with rate $\gamma$. Thus, by Theorem \ref{theorem:pdissipative},
the overall closed-loop is $2$-dominant with rate $\gamma$. Indeed, $\gamma$ 
divides the fast op-amp dynamics from the dominant two dimensional slow dynamics of the linear networks. 
We conclude the section with a numerical simulation.
The parameters of Table \ref{tab:paramSchmitt} together with 
$R_{3} = 3.3 K \Omega$, $R_{4} = 1 K \Omega$ and $C_{2} = 200 \mu F$
satisfy the conditions above, thus guarantee strict $2$-dominance the closed loop 
with rate $\gamma = 25$. The fixed point in $0$ is unstable but all trajectories remain
bounded (the constraint $x_{a} \in [-E_{2}, E_{1}]$ implies that $x_{1}$ and $x_{2}$ remain bounded),
which enforces oscillations, as shown in Figure \ref{fig:limitCycle}.
\begin{figure}[htpb]
	\centering
	\includegraphics[width=0.35\textwidth]{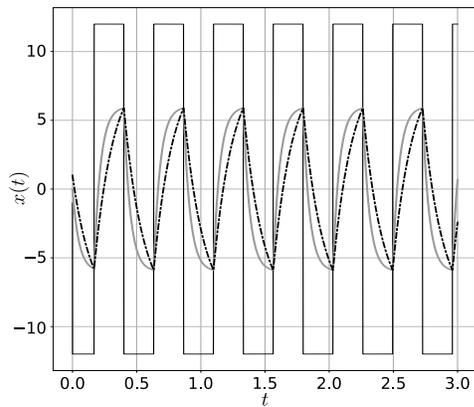}
	\caption{Trajectories of the relaxation oscillator in Figure \ref{fig:relaxOscillator}. $x_{a}(\cdot)$ -- thin/black line; 
	$x_{1}(\cdot)$ -- thick/gray line; $x_{2}(\cdot)$ -- dashed line.}
	\label{fig:limitCycle}
\end{figure}
\section{CONCLUSIONS}
\label{section:conclusions}
We extended the concept of dominance to the analysis of nonsmooth linear
complementarity systems. The extension mimics the smooth case when uniqueness of solutions is assumed.
The approach is based on  the interconnection theory of dissipativity. It opens the way to 
 the analysis of switching or oscillatory circuits with no restriction on the dimension of the state-space.
The potential of the approach was illustrated with a detailed analysis of well-known circuits based
on op-amps, predicting multistable and oscillatory behaviors, while providing margins on the circuit
parameters to enable such behaviors. 
\bibliographystyle{IEEEtran}
\bibliography{Biblio}
\end{document}